\pdfoutput=1

\documentclass[11pt]{article}
\usepackage{fullpage}
\let\subsubsection\paragraph

\usepackage[margin=10pt,font=small,labelfont=bf,labelsep=colon]{caption}
\usepackage[
	hidelinks,
	pagebackref=true,
	pdfauthor={Mruganka Kashyap and Laurent Lessard},
	pdftitle={Guaranteed Stability Margins for Decentralized Linear Quadratic Regulators}
]{hyperref}

\renewcommand*{\backrefalt}[4]{%
	\ifcase #1 \footnotesize{(Not cited.)}%
	\or        \footnotesize{(Cited on p.~#2)}%
	\else      \footnotesize{(Cited on pp.~#2)}%
	\fi}

\usepackage[bottom]{footmisc}

\usepackage{cite}                   % citation ordering / combining
\usepackage{math}                   % custom math commands
\usepackage[utf8]{inputenc}			% handle accents properly
\usepackage[english]{babel}			% handle hyphenation properly

\def\qed{\rule[0pt]{5pt}{5pt}\par\medskip}
\newcommand{\qedhere}{\hfill ~\qed}
\newenvironment{proof}{{\noindent\bf Proof.}}{\qedhere}

\usepackage{bbm}                    % blackboard bold (for fancy-looking "1" symbol)
\usepackage{textcomp}				% \degree symbol
\usepackage{graphicx}
\graphicspath{{./graphics/}}
\newcommand{\1}{\mathbf{1}}
\newcommand{\0}{\mathbf{0}}

\usepackage{tikz}
\usetikzlibrary{calc,positioning}

\usepackage{ntheorem}
\usepackage[capitalise,nameinlink]{cleveref}

\newtheorem{thm}{Theorem}

\newtheorem{lem}[thm]{Lemma}
\newtheorem{prop}[thm]{Proposition}
\newtheorem{rem}[thm]{Remark}
\newtheorem{defn}[thm]{Definition}
\newtheorem{cor}[thm]{Corollary}

\crefname{thm}{Theorem}{Theorems}
\crefname{problem}{Problem}{Problems}
\crefname{lem}{Lemma}{Lemmas}
\crefname{prop}{Proposition}{Propositions}
\crefname{rem}{Remark}{Remarks}
\crefname{defn}{Definition}{Definitions}
\crefname{cor}{Corollary}{Corollaries}
\crefname{assumption}{Assumption}{Assumptions}

\newcommand{\anc}[1]{\mathcal{A}(#1)}
\newcommand{\des}[1]{\mathcal{D}(#1)}

\usepackage{parskip}

\makeatletter
\def\thm@space@setup{%
	\thm@preskip=\parskip \thm@postskip=0pt
}
\makeatother

\begin{document}
\title{\vspace{-1cm}Guaranteed Stability Margins for Decentralized\\ Linear Quadratic Regulators}
\date{}
\author{%
	Mruganka Kashyap\thanks{M.~Kashyap is with the Department of Electrical and Computer Engineering at Northeastern University,\newline Boston, MA 02115, USA. (e-mail: kashyap.mru@northeastern.edu).}
	\and 
	Laurent Lessard\thanks{L.~Lessard is with the Department of Mechanical and Industrial Engineering at Northeastern University,\newline Boston, MA 02115, USA. (e-mail: l.lessard@northeastern.edu).}
	}
\maketitle
%%%%%%%%%%%%%%%%%%%%%%%%%%%%%%%%%%%%%%%%%%%%%%%%%%%%%%%%%%%%%%%%%%%%%%%%%%%%%%%%%%%%%%%%%%%%%%%%%

\begin{abstract}
	It is well-known that linear quadratic regulators (LQR) enjoy guaranteed stability margins, whereas linear quadratic Gaussian regulators (LQG) do not. In this letter, we consider systems and compensators defined over directed acyclic graphs. In particular, there are multiple decision-makers, each with access to a different part of the global state. In this setting, the optimal LQR compensator is dynamic, similar to classical LQG. We show that when sub-controller input costs are decoupled (but there is possible coupling between sub-controller state costs), the decentralized LQR compensator enjoys similar guaranteed stability margins to classical LQR. However, these guarantees disappear when cost coupling is introduced. 
\end{abstract}

%%%%%%%%%%%%%%%%%%%%%%%%%%%%%%%%%%%%%%%%%%%%%%%%%%%%%%%%%%%%%%%%%%%%%%%%%%%%%%%%%%%%%%%%%%%%%%%%%
\section{Introduction}\label{sec:intro}

\bigskip
Multi-agent systems with communication constraints occur naturally in engineering applications, including bilateral teleoperation systems in remote robotic surgery and unmanned aerial vehicles (UAVs). For example, a swarm of UAVs could be deployed to survey an uncharted region or to optimize geographic coverage while combating forest fires. Information transfer within the swarm could be limited due to geographic constraints such as mountains blocking line-of-sight communications between certain UAVs. 

It is known that certain decentralized information-sharing architectures lead to tractable optimal control problems~\cite{rotkowitz2005characterization,ho1972team}. One such problem is \emph{decentralized LQR} where the communication constraints have a poset-causal architecture \cite{swigart_spectral,shah2013cal}. Although this is a state-feedback problem, the optimal decentralized controller is \emph{dynamic} and has an observer-regulator structure reminiscent of output-feedback LQG regulators.

Robustness is an important aspect of controller design, because it ensures that the controller can effectively and reliably control a system in the presence of disturbances, plant uncertainty, or unmodeled dynamics. In the centralized case, LQR controllers enjoy guaranteed gain and phase margins \cite{safonov1977gain,lehtomaki1981robustness}. However, linear quadratic Gaussian (output feedback) regulators, have no robustness guarantees \cite{doyle1978guaranteed}.

The robustness properties of decentralized LQR are not immediately apparent, since decentralized LQR shares commonalities with both centralized LQR (uses state feedback), and centralized LQG (optimal controller is dynamic). To the best of our knowledge, this is an open problem.

In this letter, we show that decentralized LQR enjoys similar stability margins to classical LQR if the input matrix ($B$) and control weighting matrix ($R$) are block-diagonal. We also show via counterexample that these assumptions are necessary.

In \cref{sec:LQR,sec:decentralized_lqr} we review classical stability margins for LQR and more recent work on decentralized LQR synthesis. In \cref{sec:main} we present our main results, and in \cref{sec:discussion,sec:conclufuture} we present our counterexample and conclude.
\looseness=-1

%%%%%%%%%%%%%%%%%%%%%%%%%%%%%%%%%%%%%%%%%%%%%%%%%%%%%%%%%%%%%%%%%%%%%%%%%%%%%%%%%%%%%%%%%%%%%%%%%
\section{Classical LQR stability margins}\label{sec:LQR}

Consider the continuous-time linear time-invariant (LTI) dynamical system $\dot x = Ax + Bu$,
where $x(t)\in\R^n$ and $u(t) \in \R^m$. The linear quadratic regulator (LQR) problem is to find the causal state-feedback policy that minimizes the quadratic cost
\begin{equation}\label{eq:cost}
J = \int_0^\infty \bigl( x(t)^\tp Q x(t) + u(t)^\tp R u(t) \bigr)\,\mathrm{d}t.
\end{equation}
\begin{prop}\label{prop:lqr}
	Suppose $(A,B)$ is stabilizable, $(Q,A)$ is detectable, and $Q\succeq 0$ and $R\succ 0$. The optimal LQR policy is $u(t) = F x(t)$, where $F = -R^{-1} B^\tp X$, and $X\succeq 0$ is the unique stabilizing solution to the algebraic Riccati equation
	$A^\tp X + XA + Q - XBR^{-1} B^\tp X = 0$.
\end{prop}

We denote the optimal LQR gain from \cref{prop:lqr} using the notation $F 
	\defeq \ric(A,B,Q,R)$.
The optimal LQR controller is known to be inherently robust \cite[\S23]{hespanha2018linear}
\cite[\S14.4]{zdg}. In particular, if we define the loop gain $L(s) \defeq F(sI-A)^{-1}B$, then the \emph{Kalman inequality} holds:
\begin{equation}\label{eq:kalman_inequality}
\left(I - L(j\omega)\right)^* R \left(I - L(j\omega)\right) \succeq R	\quad\text{for all }\omega \in \R.
\end{equation}
In the single-input case, $L(j\omega)$ is a scalar and the Kalman inequality reduces to $|1-L(j\omega)|\geq 1$. This can be interpreted as the open-loop Nyquist plot of $-L$ (negative feedback) lying outside the disk centered at $(-1,0)$ with radius $1$. This implies that the LQR compensator has gain margin $\frac{1}{2}<k<\infty$ and phase margin $-60\text{\textdegree} < \phi < 60\text{\textdegree}$.

Alternatively, a sufficient condition for robust stability can be expressed in terms of the perturbation itself \cite{lehtomaki1981robustness}.

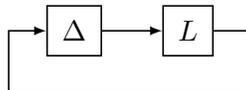
\begin{figure}[ht]
	\centering
	\begin{tikzpicture}[>=latex, node distance=8mm,semithick]
		\node[draw, inner sep=2mm] (D) {$\Delta$};
		\node[draw, inner sep=2mm, right= of D] (L) {$L$};
		\draw[->] (D) -- (L);
		\path (D.west) -- ++(-0.5,0) coordinate (t);
		\draw[->] (L.east) -- ++(0.5,0) -- ++(0,-0.8) -| (t) -- (D);
	\end{tikzpicture}
	\caption{Perturbed feedback interconnection. $L(s) \defeq F(sI-A)^{-1}B$ is the loop gain for a standard LQR feedback controller $u(t) = F x(t)$.}
	\label{fig:perturb}
\end{figure}

\begin{lem}\label{lem:lehtomaki}
Consider the setting of \cref{prop:lqr}, and let $L(s) \defeq F(sI-A)^{-1} B$ be the LQR-optimal loop gain. The interconnected system of \cref{fig:perturb} is well-posed and internally stable for all LTI systems $\Delta$ that satisfy 
\begin{equation}\label{eq:lehtomaki}
	\Delta(j\omega)^* R + R\Delta(j\omega) \succ R	\quad\text{for all }\omega \in \R.
\end{equation}
\end{lem}

\begin{proof}
Invert \eqref{eq:kalman_inequality} and apply the matrix inversion lemma, which yields
$(I+H(j\omega))^*R(I+H(j\omega)) \preceq R$,
where we defined the closed-loop map $H(s) \defeq F(sI-A-BF)^{-1}B$. This is equivalent to $\norm{R^{1/2}(I+H)R^{-1/2}}_\infty \leq 1$. Then, perform a loop-shifting transformation to \cref{fig:perturb} to obtain \cref{fig:perturb2}. Apply the small gain theorem \cite[Thm.~9.1]{zdg} to conclude that the interconnection is well-posed and internally stable for all LTI systems $\Delta$ satisfying $\norm{R^{1/2}(I-\Delta^{-1})R^{-1/2}}_\infty<1$, which is equivalent to \eqref{eq:lehtomaki}.
\end{proof}

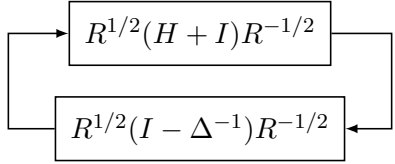
\begin{figure}[ht]
	\centering
	\begin{tikzpicture}[>=latex, node distance=4mm,semithick]
		\node[draw, inner sep=2mm] (L) {$R^{1/2}(H+I)R^{-1/2}$};
		\node[draw, inner sep=2mm, below= of L] (D) {$R^{1/2}(I-\Delta^{-1})R^{-1/2}$};
		\draw[->] (L.east) -- ++(0.8,0) |- (D);
		\draw[<-] (L.west) -- ++(-0.8,0) |- (D);
	\end{tikzpicture}
	\caption{Transformation of \cref{fig:perturb}. $H(s) \defeq F(sI-A-BF)^{-1}B$ is the closed-loop map for a standard LQR feedback controller.}
	\label{fig:perturb2}
\end{figure}

\cref{lem:lehtomaki} allows us to specialize the previous gain and phase margin results derived from the Kalman inequality to the case where each input channel is separately perturbed.

\begin{cor}\label{cor:centralized_vector_robustness}
Consider the setting of \cref{lem:lehtomaki}. Partition the input $u(t)$ into subvectors of dimension $m_1+\dots+m_N = m$. If we assume $R$ and $\Delta$ are block-diagonal and partitioned conformally to the partition of $u(t)$, i.e.,
\[
R = \bmat{R_1 & & 0\\& \ddots & \\0 & & R_N}
\quad\text{and}\quad
\Delta = \bmat{\Delta_1 & & 0\\& \ddots & \\0 & & \Delta_N},
\]
then the interconnected system of \cref{fig:perturb} is well-posed and internally stable for all independent LTI perturbations of the blocks of $u(t)$ satisfying 
$\Delta_i(j\omega)^*R_i + R_i \Delta_i(j\omega) \succ R_i$ for $i=1,\dots,N$ and for all $\omega\in\R$.
In particular, each input block independently has gain margin $\frac{1}{2}<k_i<\infty$ and phase margin $-60\text{\textdegree} < \phi_i < 60\text{\textdegree}$.
\end{cor}

In \cref{cor:centralized_vector_robustness}, the assumption that $R$ is block-diagonal is necessary. It is possible to construct systems where a non-diagonal $R$ leads to closed loops that be destabilized by arbitrarily small perturbations in a single channel \cite[Ex.~3.1]{lehtomaki1981robustness}.

Similar robustness results to \cref{lem:lehtomaki} have been derived for discrete time \cite{shaked1986guaranteed} and for the case with cross-product cost terms \cite{chung1994stability}, though these cases generally have weaker robustness guarantees.
There are also negative results; when $R$ is full, the independent perturbation result of \cref{cor:centralized_vector_robustness} no longer holds \cite[Ex.~3.1]{lehtomaki1981robustness}. Finally, there are no guaranteed stability margins for LQG compensators \cite{doyle1978guaranteed}.

\section{Decentralized LQR control}\label{sec:decentralized_lqr}

We consider the problem setting studied in \cite{shah2013cal,swigart_spectral}, which is an LQR problem structured according to a directed acyclic graph (DAG).
Specifically, we assume the setting in \cref{prop:lqr}, but we partition the state as $x = \bmat{x_1^\tp & \cdots & x_N^\tp}^\tp$ and similarly for the input $u$. We also partition $A$ and $B$ as $N\times N$ block matrices conforming to the partitions of $x$ and $u$.

There is an underlying DAG on the nodes $1,\dots,N$, which are assumed to be ordered according to the partial ordering of the DAG. The matrices $A$ and $B$ have a block-sparsity pattern that conforms to the adjacency matrix of the transitive closure of the DAG. Consider for example the 4-node DAG in \cref{fig:dag_example}.

\begin{figure}[ht]
	\centering
	\begin{minipage}{0.4\linewidth}
		\centering
		\begin{tikzpicture}[semithick]
			\tikzstyle{vertex} = [draw,circle,inner sep=1mm,font=\small]
			\tikzstyle{arr} = [>=latex,->]
			\def\y{0.6}
			\def\x{1.0}
			\node[vertex] (1) at (0,0) {$1$};
			\node[vertex] (2) at (-\x,-\y) {$2$};
			\node[vertex] (3) at (\x,-\y) {$3$};
			\node[vertex] (4) at (0,-2*\y) {$4$};
			\draw[arr] (1) -- (2);
			\draw[arr] (1) -- (3);
			\draw[arr] (2) -- (4);
			\draw[arr] (3) -- (4);
		\end{tikzpicture}
	\end{minipage}%
	\begin{minipage}{0.4\linewidth}
		\[
		S = \bmat{1 & 0 & 0 & 0\\
			1 & 1 & 0 & 0 \\
			1 & 0 & 1 & 0 \\
			1 & 1 & 1 & 1}
		\]
	\end{minipage}
	\caption{Example of a 4-node directed acyclic graph (DAG), the adjacency matrix of its transitive closure is $S$, shown on the right.}
	\label{fig:dag_example}
\end{figure}
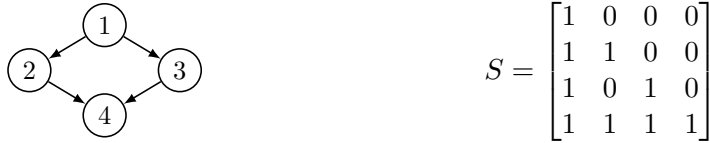

\noindent The associated dynamical system would have the structure
\begin{small}
\[
\bmat{\dot{x}_1\\\dot{x}_2\\\dot{x}_3\\\dot{x}_4}
=
\addtolength{\arraycolsep}{-1.2mm}\bmat{A_{11} & 0 & 0 & 0 \\ A_{21} & A_{22} & 0 & 0 \\ A_{31} & 0 & A_{33} & 0 \\ A_{41} & A_{42} & A_{43} & A_{44}}\!\!
\bmat{x_1 \\ x_2 \\ x_3 \\ x_4}
+
\bmat{B_{11} & 0 & 0 & 0 \\ B_{21} & B_{22} & 0 & 0 \\ B_{31} & 0 & B_{33} & 0 \\ B_{41} & B_{42} & B_{43} & B_{44}}\!\!
\bmat{u_1 \\ u_2 \\ u_3 \\ u_4}	
\]\end{small}%
There are no assumptions on the cost matrices, so all states and inputs may be coupled through $Q$ and $R$, respectively.

\begin{defn}
The ancestors of node $i$, denoted $\anc{i}$, is the set of all nodes $j$ for which there exists a directed path from $j$ to $i$, including node $i$. Similarly, the descendants of node $i$, denoted $\des{i}$, is the set of all nodes $j$ for which there exists a directed path from $i$ to $j$, including $i$. We also use these sets as a matrix subscripts to indicate the submatrix formed by selecting the corresponding block rows and columns.
\end{defn}

For the example of \cref{fig:dag_example}, we have $\des{2} = \{2,4\}$ and $\anc{3} = \{1,3\}$, which defines the block submatrices
\[
A_{\des{2}} = \bmat{A_{22} & 0 \\ A_{24} & A_{44}}
\quad\text{and}\quad
B_{\anc{3}} = \bmat{B_{11} & 0 \\ B_{31} & B_{33}}.
\]

What makes the problem \emph{decentralized} is that each $u_i$ only has access to the past history of the ancestors of node $i$. For the example of \cref{fig:dag_example}, this means the $u_i$ take the form
\begin{align*}
u_1 &= \mathcal{K}_1(x_1), &
u_2 &= \mathcal{K}_2(x_1,x_2), \\
u_3 &= \mathcal{K}_3(x_1,x_3), &
u_4 &= \mathcal{K}_4(x_1,x_2,x_3,x_4),
\end{align*}
where the $\mathcal{K}_i$ are causal maps.
In general, decentralized problems with LQG assumptions need not have linear optimal controllers \cite{witsenhausen1968counterexample}. However, when the plant and controller are structured according to a DAG as above, the optimal controller is linear \cite{ho1972team} and finding the optimal linear controller may be cast as a convex optimization problem \cite{rotkowitz2002decentralized}.

Explicit closed-form solutions have been obtained for this decentralized LQR problem using a state-space approach~\cite{kim2015explicit,swigart_spectral} and poset-based approach~\cite{shah2013cal}. Similar explicit solutions exist for LQG (output-feedback) versions of this problem \cite{lessard2015optimal,kim2015explicit,tanaka2014optimal,kashyap2019explicit} and also with time delays \cite{kashyap2020agent,lamperski2015optimal}.

The optimal controller for the decentralized LQR problem described above has the following structure \cite{swigart_spectral}.

\begin{prop}\label{prop:dec_lqr}
	Consider the decentralized LQR problem. Suppose $(A_i,B_i)$ is stabilizable for $i=1,\dots,N$ and $(Q,A)$ is detectable.
	Let $F_i \defeq \ric(A_{\des{i}},B_{\des{i}},Q_{\des{i}},R_{\des{i}})$.
	The optimal decentralized LQR controller has closed-loop dynamics and associated optimal policy given by
	\begin{align*}
		\begin{aligned}
		\dot{\xi}_i &= (A_{\des{i}} + B_{\des{i}} F_i)\xi_i\\
		u_i &= \sum_{j \in \anc{i}} I_{i,\des{j}}F_j\xi_j
		\end{aligned} & & \text{for }i=1,\dots,N
	\end{align*}
	where $I_{i,\des{j}}$ is the block-row of the identity matrix $I_{\des{j}}$ associated with node $i$.
	% where $I_{i,\des{j}}F_j$ selects the block-row of $F_j$ corresponding to $i$ if $i\in \des{j}$, and is zero otherwise.

	If we include zero-mean process noise in the plant dynamics that is independent between the different nodes of the DAG, then
	$\xi_i = \ee\bigl( x_{\des{i}} \suchthat x_{\anc{i}} \bigr) - \ee\bigl( x_{\des{i}} \suchthat x_{\anc{i}\setminus \{i\}} \bigr)$, so $\xi_i$ is an estimation correction in updating the estimate of the descendants once the current node $i$ is included.
\end{prop}

The optimal decentralized controller from \cref{prop:dec_lqr} is linear, but unlike the classical centralized case in \cref{prop:lqr}, it is also \emph{dynamic}. The decentralized LQR controller bears a resemblance to the optimal LQG controller because its states are estimates of plant states. The main difference is that the strict descendants of node $i$ are not observable, so rather than using an observer such as a Kalman filter, the state estimates are formed via \emph{prediction} \cite[\S IV.D]{shah2011optimal}.

%%%%%%%%%%%%%%%%%%%%%%%%%%%%%%%%%%%%%%%%%%%%%%%%%%%%%%%%%%%%%%%%%%%%%%%%%%%%%%%%%%%%%%%%%%%%%%%%%
\section{Main Results}\label{sec:main}

For the optimal decentralized LQR controller described in \cref{prop:dec_lqr}, there is no large Kalman inequality of the form \eqref{eq:kalman_inequality}. Instead, we have $N$ separate Kalman inequalities
\begin{equation}\label{eq:kalman_inequality_decent}
	\left(I - L_i(j\omega)\right)^* R_{\des{i}} \left(I - L_i(j\omega)\right) \succeq R_{\des{i}}\qquad\text{for all }\omega \in \R,
\end{equation}
corresponding to the $N$ separate centralized LQR sub-problems that make up the optimal decentralized controller.

Consequently, there is no apparent way to leverage the small gain theorem as in the proof of \cref{lem:lehtomaki}. Instead, we show that if we assume $B$ and $R$ are block-diagonal, we can prove a result similar to \cref{cor:centralized_vector_robustness} for block-diagonal perturbations. 

\begin{thm}\label{thm:main}
	Consider the decentralized LQR problem and its optimal controller, described in \cref{sec:decentralized_lqr,prop:dec_lqr}, respectively, and let $L_\textup{dec}$ be the optimal loop gain. 
	
	Further suppose that $R$ and $B$ are block-diagonal with block sizes corresponding to the partitions of $x(t)$ and $u(t)$.
	The interconnected system of \cref{fig:perturb_dec} is well-posed and internally stable for all independent LTI perturbations of the blocks of $u(t)$ satisfying the following for all $i=1,\dots,N$. 
	\begin{equation}\label{eq:new_lehto}
		\Delta_i(j\omega)^* R_i + R_i\Delta_i(j\omega) \succ R_i\quad\text{for all }\omega \in \R.
	\end{equation}
\end{thm}

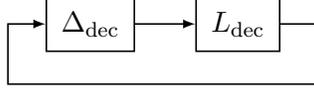
\begin{figure}[ht]
	\centering
	\begin{tikzpicture}[>=latex, node distance=8mm,semithick]
		\node[draw, inner sep=2mm] (D) {$\Delta_\textup{dec}$};
		\node[draw, inner sep=2mm, right= of D] (L) {$L_\textup{dec}$};
		\draw[->] (D) -- (L);
		\path (D.west) -- ++(-0.5,0) coordinate (t);
		\draw[->] (L.east) -- ++(0.5,0) -- ++(0,-0.8) -| (t) -- (D);
	\end{tikzpicture}
	\caption{Perturbed feedback interconnection. $L_\textup{dec}$ is the loop gain for the optimal decentralized LQR feedback controller described in \cref{prop:dec_lqr} and $\Delta_\textup{dec} = \diag\{\Delta_i\}$ is a block-diagonal LTI perturbation.}
	\label{fig:perturb_dec}
\end{figure}

\begin{rem}
	\cref{thm:main} looks similar to \cref{cor:centralized_vector_robustness}, but $L_\textup{dec}$ is now the more complicated loop gain for the optimal \emph{decentralized} LQR controller. Unlike \cref{cor:centralized_vector_robustness}, \cref{thm:main} makes the additional assumptions that $B$ and $R$ are block diagonal. In \cref{sec:discussion}, we show that these assumptions are necessary, but we argue that they are not restrictive in many cases of practical interest.
\end{rem}

\begin{proof}
We take an approach similar to the proof of \cref{lem:lehtomaki}, except we use a more general version of the small gain theorem for structured uncertainty, and additional steps are required to combine the $N$ separate Kalman inequalities into something we can use. Start by rewriting the closed-loop map of the optimal decentralized LQR controller from \cref{prop:dec_lqr} as:
\[
H_\textup{dec} = \1^\tp \bar{F} (sI - \bar A - \bar B \bar F)^{-1} \hat B
\]
where we defined:
\begin{align*}
\1^\tp & \defeq \bmat{I_{1,\des{1}} & \cdots & I_{N,\des{N}}} \\
\bar A & \defeq \diag\{ A_{\des{i}} \} \\
\bar F & \defeq \diag\{ F_i \} \\
\bar B & \defeq \diag\{ B_{\des{i}} \} \\
\hat B & \defeq \bmat{ e_1e_1^\tp B_{\des{1}} \\ \vdots \\ e_Ne_N^\tp B_{\des{N}} },
\end{align*}
where $e_i$ is the $i$-th column of the identity matrix of size $n$. Since $B$ is block-diagonal, we have $\hat B = \diag\{ B_{\des{i}}e_i \}$. So we can rewrite the closed-loop map as
$
H_\textup{dec} = \1^\tp \bar H E
$, where we defined:
\begin{gather*}
\bar H \defeq \diag\{H_i\}, \quad
E \defeq \diag\{e_i\}, \quad
\bar R \defeq \diag\{R_{\des{i}}\}, \\
H_i \defeq F_i (sI-A_{\des{i}}-B_{\des{i}} F_i)^{-1} B_{\des{i}}.
\end{gather*}
Note that $H_i$ is the closed-loop map for the separate LQR problem associated with $\des{i}$ defined in \cref{prop:dec_lqr}. 

Now perform the same loop-shifting transformation as in the proof of \cref{lem:lehtomaki} to \cref{fig:perturb_dec} to obtain \cref{fig:perturb_dec2}. Since $R$ and $\Delta_\textup{dec}$ are block-diagonal, the uncertainty block in \cref{fig:perturb_dec2} is also block-diagonal. Our goal is to apply the structured small gain theorem \cite[Thm.~11.8]{zdg}, which is a generalization of the small gain theorem that applies when the uncertainty is structured.

\begin{figure}[ht]
	\centering
	\begin{tikzpicture}[>=latex, node distance=4mm,semithick]
		\node[draw, inner sep=2mm] (L) {$R^{1/2}(\1^\tp \bar H E+I)R^{-1/2}$};
		\node[draw, inner sep=2mm, below= of L] (D) {$R^{1/2}(I-\Delta_\textup{dec}^{-1})R^{-1/2}$};
		\draw[->] (L.east) -- ++(0.8,0) |- (D);
		\draw[<-] (L.west) -- ++(-0.8,0) |- (D);
	\end{tikzpicture}
	\caption{Transformation of \cref{fig:perturb_dec}. $\bar H \defeq \diag\{H_i\}$ is the block-diagonal concatenation of the closed-loop maps associated with the $N$ centralized LQR sub-problems that make up the optimal decentralized LQR controller.}
	\label{fig:perturb_dec2}
\end{figure}
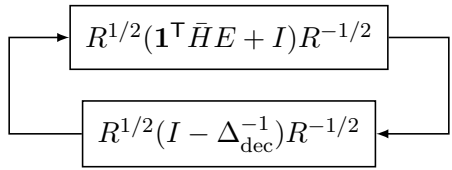

To this end, we state an intermediate lemma, which relates the structured singular value of the optimal closed-loop map to the $N$ separate Kalman inequalities \eqref{eq:kalman_inequality_decent}.

\begin{lem}\label{lem:mu}
Consider the setting of \cref{thm:main}, where $\bar H$, $H_i$, $E$, and $\1$ are defined as above. The following inequality holds:
\begin{equation*}
\sup_{\omega\in\R} \, \mu_\Delta\Bigl( R^{1/2}(\1^\tp \bar H(j\omega) E + I)R^{-1/2} \Bigr)
\leq\,\, \max_{1\leq i \leq N}\, \normm{R_i^{1/2}e_i^\tp \bigl(I+H_i\bigr)e_i R_i^{-1/2}}_\infty,	
\end{equation*}
where $\mu_\Delta(\cdot)$ denotes the structured singular value corresponding to the block-diagonal structure of $\Delta_\textup{dec}$.
\end{lem}
\begin{proof}
Let $M \defeq R^{1/2}(\1^\tp \bar H(j\omega) E + I)R^{-1/2}$. Since the plant and controller each have transfer functions structured according to the adjacency matrix $S$ of the transitive closure of the associated DAG, they form an algebra. Consequently, all products, inverses, and linear fractional transformations preserve the structure, and in particular, so does the closed-loop map $H_\textup{dec}$. Therefore, $M$ has a block-sparsity structure conforming to $S$. Since the nodes are assumed to be ordered according to the partial ordering of the DAG, $S$ is lower-triangular and so $M$ is block-lower triangular. 

Let $\mathbf{\Delta} \defeq \set{\diag\{\Delta_i\} }{\Delta_i\in\C^{m_i\times m_i}}$. For any $\Delta\in\mathbf{\Delta}$,
\[
	\det(I-M\Delta) = \prod_{i=1}^N \det(I-M_{ii}\Delta_i)
\]
and we can simplify $M_{ii}$ based on the definition as
\begin{align}\label{Miiineq}
	M_{ii} &= e_i^\tp \left( R^{1/2}(\1^\tp \bar H(j\omega) E + I)R^{-1/2} \right) e_i \notag \\
	&= R_i^{1/2} e_i^\tp ( \1^\tp \bar H(j\omega) E + I)e_i R_i^{-1/2} \notag \\
	&= R_i^{1/2} e_i^\tp (H_i(j\omega) + I)e_i R_i^{-1/2}.
\end{align}
By the definition of the structured singular value, 
\begin{align*}
\mu_\Delta(M)
&= \frac{1}{\min\set{\norm{\Delta}}{\det(I-M\Delta)=0,\Delta\in\mathbf{\Delta}}} \\
&= \frac{1}{\min\set{\norm{\Delta}}{\det(I-M_{ii}\Delta_i)=0 \text{ for some }i}} \\
&\leq \frac{1}{\min_{i}\min\set{\norm{\Delta_i}}{\det(I-M_{ii}\Delta_i)=0}} \\
&= \max_{1\leq i \leq N}\,\norm{M_{ii}}.
\end{align*}
The last step follows from the fact that $\mu_{\Delta_i}(M_{ii}) = \norm{M_{ii}}$ because $\Delta_i$ is unstructured. Substituting in $M_{ii}$ from \eqref{Miiineq} and taking the supremum over $\omega\in\R$ completes the proof.
\end{proof}
\bigskip

Inverting the Kalman inequalities in \eqref{eq:kalman_inequality_decent} and converting them into $\Hinf$ norms as in the proof of \cref{lem:lehtomaki}, we obtain 
\[
	\normm{R_{\des{i}}^{1/2}\bigl(I + H_i\bigr)R_{\des{i}}^{-1/2}}_\infty \leq 1
	\quad\text{for }i=1,\dots,N.
\]
Since for any matrix $M\in\C^{p\times q}$, the (spectral) norm of $M$ is lower-bounded by the norm of any submatrix of $M$, have a similar inequality for $\Hinf$ norms, and together with the fact that $R$ is block-diagonal, we deduce that
\[
	\normm{R_i^{1/2}e_i^\tp \bigl(I+H_i\bigr)e_i R_i^{-1/2}}_\infty
	\leq
	\normm{R_{\des{i}}^{1/2}\bigl(I + H_i\bigr)R_{\des{i}}^{-1/2}}_\infty	
\]
The two above inequalities together with \cref{lem:mu} imply that
\[
\sup_{\omega\in\R} \, \mu_\Delta\Bigl( R^{1/2}(\1^\tp \bar H(j\omega) E + I)R^{-1/2} \Bigr)\\
\leq 1.
\]
We can now apply the structured small gain theorem \cite[Thm.~11.8]{zdg} and conclude that the interconnection of \cref{fig:perturb_dec2} is well-posed and stable whenever $\Delta_\textup{dec} = \diag\{\Delta_i\}$ satisfy
\[
\normm{R^{1/2}(I-\Delta_\textup{dec}^{-1})R^{-1/2} }_\infty < 1.	
\]
Due to the block-diagonal structure of the uncertainty, this is equivalent to
\[
\normm{R_i^{1/2}(I-\Delta_i^{-1})R_i^{-1/2} }_\infty < 1\quad\text{for }i=1,\dots,N	
\]
which is equivalent to \eqref{eq:new_lehto}.
\end{proof}
\bigskip

Equipped with \cref{thm:main}, we can specialize the decentralized LQR robustness result to the case where each input channel is perturbed using either a pure gain or a pure phase shift. This leads us to a decentralized version of \cref{cor:centralized_vector_robustness}.

\begin{cor}\label{cor:main}
	Consider the decentralized LQR setting of \cref{thm:main}. 
	Each input $u_i(t)$ independently has gain margin $\frac{1}{2}<k_i<\infty$ and phase margin $-60\text{\textdegree} < \phi_i < 60\text{\textdegree}$.
\end{cor}

%%%%%%%%%%%%%%%%%%%%%%%%%%%%%%%%%%%%%%%%%%%%%%%%%%%%%%%%%%%%%%%%%%%%%%%%%%%%%%%%%%%%%%%%%%%%%%%%%
\section{Discussion}\label{sec:diss}\label{sec:discussion}

\cref{thm:main} provides conditions for the robust stability of the optimal decentralized linear quadratic regulator, under the additional assumptions that $B$ and $R$ are block-diagonal and different perturbations are applied to each input $u_i$.

The assumption that $B$ and $R$ are block-diagonal is critical. We will demonstrate using a simple numerical example that the gain margin $\frac{1}{2}<k_i<\infty$ established in \cref{cor:main} no longer applies when either $B$ or $R$ is not block-diagonal.

Consider a two-node DAG with graph $1 \to 2$ and global plant dynamics given by
\begin{equation}\label{eq:ex}
	\dot{x}=\begin{bmatrix}1&0\\
	1&1\end{bmatrix}x+\begin{bmatrix}1&0\\
	\beta&1\end{bmatrix}u,
	\end{equation}
cost matrices $Q=\sbmat{3 & 1 \\ 1 & 3}$ and $R=\sbmat{100 & \rho \\ \rho & 100}$. We use the perturbation $\Delta_\textup{dec} = \sbmat{k & 0 \\ 0 & 1}$ with $k\in\R$, so node $1$ is perturbed by a static scalar gain while node $2$ remains unperturbed. The perturbed closed-loop matrix is given by
\[
{A}_{\text{CL}}\defeq\bmat{
	1+k F^{11}_1& k F^{12}_1 & 0 \\
	1+\beta F^{11}_1+F^{21}_1& 1+\beta F^{12}_1+F^{22}_1& 0 \\
	1+k \beta F^{11}_1+ F^{21}_1& F^{22}_1-F_2+k \beta F^{12}_1 & 1+F_2
}
\]
where $F^{ij}_1$ and $F_2$ are given by
\begin{align*}
	\addtolength{\arraycolsep}{-0.5mm}\bmat{F^{11}_1 & F^{12}_1 \\ F^{21}_1 & F^{22}_1}
&= \ric\biggl(\addtolength{\arraycolsep}{-0.5mm}\bmat{1 & 0 \\ 1 & 1}, \bmat{1 & 0 \\ \beta & 1}, \bmat{3 & 1 \\ 1 & 3}, \bmat{100& \rho \\ \rho &100} \biggr) \\
F_2 &= \ric( 1, 1, 3, 100) \approx -2.0149.
\end{align*}

The gain margin of input $u_1$ is the range of values of $k$ for which $A_\text{CL}$ is Hurwitz.

We ran two experiments. First, we assumed a diagonal $R$ and triangular $B$, so we fixed $\rho=0$ and varied $\beta$. \cref{fig:graph_decent_example} (top) shows a plot of the pairs $(\beta,k)$ for which $A_\text{CL}$ is Hurwitz (shaded in blue). When $\beta=0$, we confirm the result of \cref{cor:main}; the system is stable for $\frac{1}{2}<k<\infty$, which corresponds to $k > -6\,\text{dB}$ on the plot. But when $\beta\neq 0$, violating the requirement that $B$ be block-diagonal, we observe a severe deterioration in the gain margin.

For the second experiment, we assumed a full $R$ but diagonal $B$, so we fixed $\beta=0$ and varied $\rho$. \cref{fig:graph_decent_example} (bottom) shows a plot of the pairs $(\rho,k)$ for which $A_\text{CL}$ is Hurwitz (shaded in blue). As in the previous example, we confirm the result of \cref{cor:main} when $\rho=0$, but we observe deterioration for some nonzero choices of $\rho$.

The matrices $B$ and $R$ are block-diagonal in many cases of practical interest. For example, consider multi-agent systems, such as drones flying in formation or a platoon of vehicles. In these cases, each control input affects a separate agent, so $B$ is block-diagonal. Also, the total input cost is typically the sum of input costs for each agent, with no coupling. So $R$ is block-diagonal as well.

%%%%%%%%%%%%%%%%%%%%%%%%%%%%%%%%%%%%%%%%%%%%%%%%%%%%%%%%%%%%%%%%%%%%%%%%%%%%%%%%%%%%%%%%%%%%%%%%%
\section{Conclusion}
\label{sec:conclufuture}
We studied the robustness of optimal decentralized LQR controllers when the plant and controller are structured according to a directed acyclic graph. Specifically, we established that when the $B$ and $R$ matrices are block-diagonal and different LTI perturbations are applied to each input, the controlled system enjoys the same stability margins as in the classical (centralized) LQR case. This is an interesting result because the optimal decentralized LQR controller is dynamic, much like an output-feedback LQG controller, yet LQG controllers have no stability margins.

While this letter only studied the case of LTI perturbations, our approach can be generalized to nonlinear input perturbations, analogous to the results obtained in~\cite{safonov1977gain}.

\begin{figure}[!ht]
	\centering
	\includegraphics[scale=1.2]{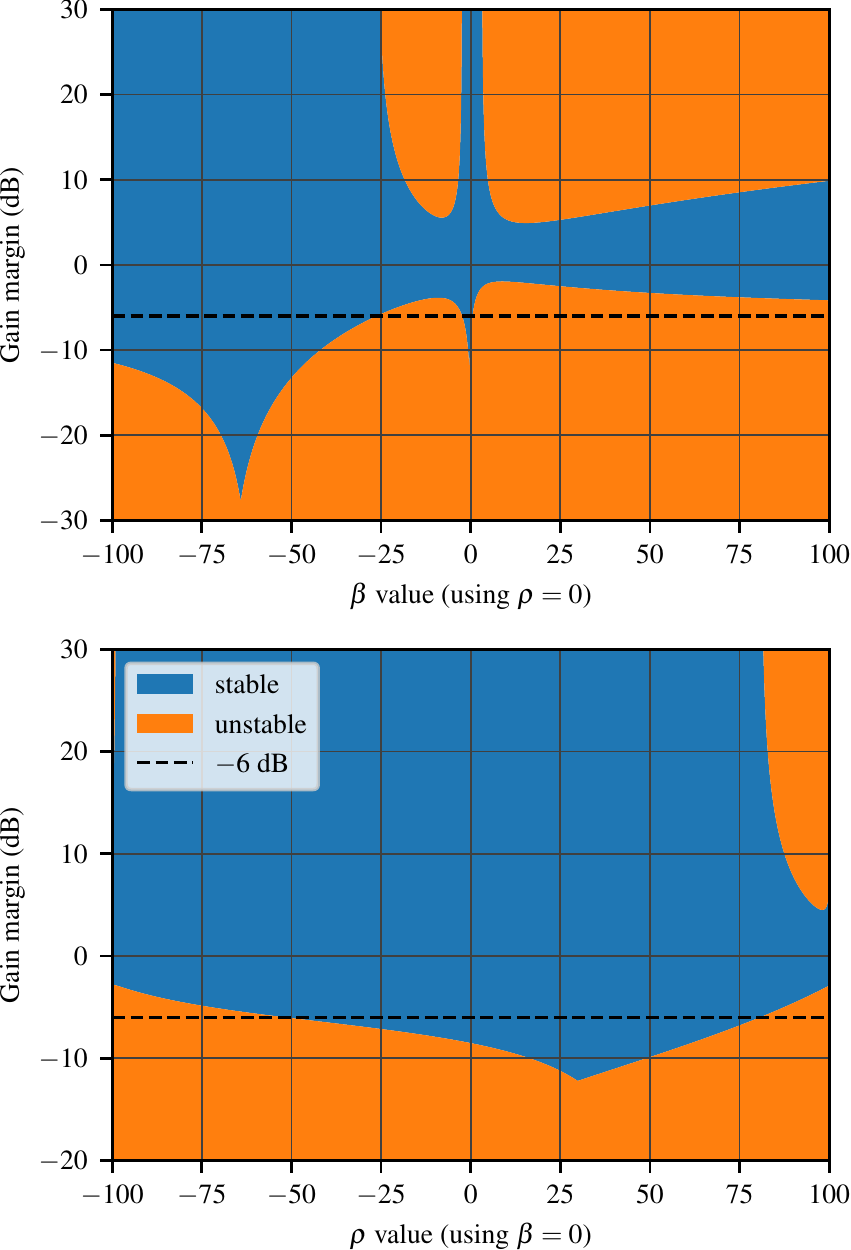}
	\caption{Stability margins for the decentralized LQR example with the dynamics of \cref{eq:ex}. The input $u_1$ is perturbed by a factor of $k$. The top panel uses $\rho=0$ (diagonal $R$) and the blue region shows the $(\beta,k)$ that yield a stable closed loop, with $k$ expressed in decibels (dB). The bottom panel uses $\beta=0$ (diagonal $B$) and the blue region shows the $(\rho,k)$ that yield a stable closed loop.
	When $B$ and $R$ are block-diagonal ($\rho=\beta=0$), we recover \cref{cor:main}, which ensures a  gain margin $\frac{1}{2}<k<\infty$. In other words, $k > -6\,\text{dB}$.}
	\label{fig:graph_decent_example} 
\end{figure}

%%%%%%%%%%%%%%%%%%%%%%%%%%%%%%%%%%%%%%%%%%%%%%%%%%%%%%%%%%%%%%%%%%%%%%%%%%%%%%%%%%%%%%%%%%%%%%%%%

\newpage\
\newpage
\bibliographystyle{abbrv}
{\small \bibliography{rdlqr}}

\begin{thebibliography}{10}

\bibitem{chung1994stability}
D.~Chung, T.~Kang, and J.~G. Lee.
\newblock Stability robustness of {LQ} optimal regulators for the performance
  index with cross-product terms.
\newblock {\em IEEE Transactions on Automatic Control}, 39(8):1698--1702, 1994.

\bibitem{doyle1978guaranteed}
J.~C. Doyle.
\newblock Guaranteed margins for {LQG} regulators.
\newblock {\em IEEE Transactions on Automatic Control}, 23(4):756--757, 1978.

\bibitem{hespanha2018linear}
J.~P. Hespanha.
\newblock {\em Linear systems theory, second edition}.
\newblock Princeton university press, 2018.

\bibitem{ho1972team}
Y.-C. Ho and K.~Chu.
\newblock Information structure in dynamic multi-person control problems.
\newblock {\em Automatica}, 10(4):341--351, 1974.

\bibitem{kashyap2019explicit}
M.~Kashyap and L.~Lessard.
\newblock {Explicit agent-level optimal cooperative controllers for dynamically
  decoupled systems with output feedback}.
\newblock In {\em IEEE Conference on Decision and Control}, pages 8254--8259,
  2019.

\bibitem{kashyap2020agent}
M.~Kashyap and L.~Lessard.
\newblock Agent-level optimal {LQG} control of dynamically decoupled systems
  with processing delays.
\newblock In {\em IEEE Conference on Decision and Control}, pages 5980--5985,
  2020.

\bibitem{kim2015explicit}
J.-H. Kim and S.~Lall.
\newblock Explicit solutions to separable problems in optimal cooperative
  control.
\newblock {\em IEEE Transactions on Automatic Control}, 60(5):1304--1319, 2015.

\bibitem{lamperski2015optimal}
A.~Lamperski and L.~Lessard.
\newblock Optimal decentralized state-feedback control with sparsity and
  delays.
\newblock {\em Automatica}, 58:143--151, 2015.

\bibitem{lehtomaki1981robustness}
N.~Lehtomaki, N.~Sandell, and M.~Athans.
\newblock Robustness results in linear-quadratic {G}aussian based multivariable
  control designs.
\newblock {\em IEEE Transactions on Automatic Control}, 26(1):75--93, 1981.

\bibitem{lessard2015optimal}
L.~Lessard and S.~Lall.
\newblock Optimal control of two-player systems with output feedback.
\newblock {\em IEEE Transactions on Automatic Control}, 60(8):2129--2144, 2015.

\bibitem{rotkowitz2002decentralized}
M.~Rotkowitz and S.~Lall.
\newblock Decentralized control information structures preserved under
  feedback.
\newblock In {\em IEEE Conference on Decision and Control}, volume~1, pages
  569--575, 2002.

\bibitem{rotkowitz2005characterization}
M.~Rotkowitz and S.~Lall.
\newblock A characterization of convex problems in decentralized control.
\newblock {\em IEEE Transactions on Automatic Control}, 50(12):1984--1996,
  2005.

\bibitem{safonov1977gain}
M.~Safonov and M.~Athans.
\newblock Gain and phase margin for multiloop {LQG} regulators.
\newblock {\em IEEE Transactions on Automatic Control}, 22(2):173--179, 1977.

\bibitem{shah2011optimal}
P.~Shah and P.~A. Parrilo.
\newblock An optimal controller architecture for poset-causal systems.
\newblock In {\em IEEE Conference on Decision and Control}, pages 5522--5528,
  2011.

\bibitem{shah2013cal}
P.~Shah and P.~A. Parrilo.
\newblock $\mathcal{H}_2$-optimal decentralized control over posets: A
  state-space solution for state-feedback.
\newblock {\em IEEE Transactions on Automatic Control}, 58(12):3084--3096,
  2013.

\bibitem{shaked1986guaranteed}
U.~Shaked.
\newblock Guaranteed stability margins for the discrete-time linear quadratic
  optimal regulator.
\newblock {\em IEEE Transactions on Automatic Control}, 31(2):162--165, 1986.

\bibitem{swigart_spectral}
J.~Swigart and S.~Lall.
\newblock Optimal controller synthesis for decentralized systems over graphs
  via spectral factorization.
\newblock {\em IEEE Transactions on Automatic Control}, 59(9):2311--2323, 2014.

\bibitem{tanaka2014optimal}
T.~Tanaka and P.~A. Parrilo.
\newblock Optimal output feedback architecture for triangular {LQG} problems.
\newblock In {\em American Control Conference}, pages 5730--5735, 2014.

\bibitem{witsenhausen1968counterexample}
H.~S. Witsenhausen.
\newblock A counterexample in stochastic optimum control.
\newblock {\em SIAM Journal on Control}, 6(1):131--147, 1968.

\bibitem{zdg}
K.~Zhou, J.~C. Doyle, and K.~Glover.
\newblock {\em Robust and optimal control}.
\newblock Prentice-Hall, Inc., 1996.

\end{thebibliography}

\end{document}